\documentclass[11pt,letterpaper]{article}
\usepackage[margin=1.1in]{geometry}
\usepackage{amsmath,amsfonts,amssymb,bbm,amsthm} 
\usepackage{natbib}
\usepackage[ruled,vlined,norelsize]{algorithm2e}

\newtheorem{claim}{Claim}
\newtheorem{lemma}[claim]{Lemma}
\newtheorem{theorem}[claim]{Theorem}

\newcommand{\argmax}{\text{argmax}}
\newcommand{\prob}[1]{\text{Pr}\left [#1\right]}
\newcommand{\ex}[2]{\mathbb{E}_{#1}\left[#2\right]}
\newcommand{\condex}[3]{\ex{#1}{#2\,|\,#3}}

\begin{document}
\title{An Improved Welfare Guarantee for First Price Auctions\footnote{The authors would like to thank Vasilis Syrgkanis, Jason Hartline, and Brendan Lucier for helpful discussions.}}
\author{Darrell Hoy\footnote{Tremor Technologies} \and Sam Taggart\footnote{Oberlin College} \and Zihe Wang\footnote{Shanghai University of Finance and Economics}}
\maketitle

\begin{abstract}
	This paper proves that the welfare of the first price auction in Bayes-Nash equilibrium is at least a $.743$-fraction of the welfare of the optimal mechanism assuming agents' values are independently distributed. The previous best bound was $1-1/e\approx.63$, derived in \cite{ST13} using smoothness, the standard technique for reasoning about welfare of games in equilibrium. In the worst known example (from \cite{HHT14}), the first price auction achieves a $\approx.869$-fraction of the optimal welfare, far better than the theoretical guarantee. Despite this large gap, it was unclear whether the $1-1/e\approx.63$ bound was tight. We prove that it is not. Our analysis eschews smoothness, and instead uses the independence assumption on agents' value distributions to give a more careful accounting of the welfare contribution of agents who win despite not having the highest value.
\end{abstract}

\newcommand{\ihigh}{\mathcal E_i(v_i)}
\newcommand{\utilhigh}{u_i(v_i\,|\,\ihigh)}
\newcommand{\allochigh}{x_i(v_i\,|\,\ihigh)}
\newcommand{\bidutilhigh}{\tilde u_i(b\,|\,\ihigh)}
\newcommand{\bidallochigh}{\tilde x_i(b\,|\,\ihigh)}
\newcommand{\bidutil}{\tilde u_i(b)}
\newcommand{\bidalloc}{\tilde x_i(b)}
\newcommand{\vjlb}{\phi(v_i,q_i,b_i,b_j)}
\newcommand{\vistarlb}{\underline v(v_i,q_i,b_i,b_{i^*})}
\newcommand{\utilqhigh}{u_i(v_i(q_i)\,|\,\mathcal E_i(v_i(q_i)))}

\section{Introduction}

There has been a wealth of recent progress in mechanism design understanding the performance of resource allocation mechanisms through worst-case analysis. Such results, often termed \emph{price of anarchy} results, typically bound the expected welfare in equilibrium of a mechanism without requiring the theorist to solve analytically for the equilibria being studied. 

The first price auction is a canonical example of the successes of such an approach. It has been known since \cite{ST13} that the first-price auction is guaranteed in equilibrium to achieve at least a $\tfrac{e-1}{e}$ fraction of optimal social welfare. This bound was shown to be tight in the case of Bayes-Nash equilibrium when bidders' values are correlated \citep{S14}, as well as in the simultaneous composition of multiple first-price auctions with submodular and subadditive valuations\citep{Christodoulou2016TightBF}. 

However, in the most commonly studied setting with independently distributed values over a single good, there has been a persistent gap between this welfare guarantee and the worst known example of \citet{HHT14}, in which the social welfare is a $.869$-fraction of the optimal social welfare. Despite the prevalence of the first-price auction format, the salience of the independent values assumption, and the ubiquity of the techniques used to prove the existing $\tfrac{e-1}{e}$ bound, it was not clear whether this bound was tight in the independent values setting.

In this paper, we prove the following theorem:
\begin{theorem}
	\label{thm:main}
	The expected welfare of any Bayes-Nash equilibrium of the first-price auction with independently-distributed bidders is at least a $.743$-fraction of the optimal welfare.
\end{theorem}

This improves on previous best bound of $\tfrac{e-1}{e}\approx .63$ discussed above.

\subsection{Approach}
The standard approach to proving worst-case bounds on the welfare of auctions is via \emph{smoothness}, a technique formalized in \cite{R09}, and developed for mechanisms in \cite{ST13}. Using a deviation-based argument, smoothness breaks the expected welfare of a mechanism into two parts: the contribution from agents' utility and the revenue of the auction. A welfare guarantee follows from deriving a tradeoff lowerbounding the sum of these two quantities

To prove Theorem~\ref{thm:main}, we instead consider instead a different pair of quantities: the welfare from ``rightful'' winners (that is, winners who win and have the highest value) and that from those who the mechanism improperly allocates (that is, agents who win despite not having the highest value). One way to derive a tradeoff between these two quantities would be to lower bound the value of an improper winner by their payments made to the mechanism. The alternative derivation of the $\tfrac{e-1}{e}$ bound in \cite{HHT14} implicitly follows this approach.

To prove a guarantee stronger than $\tfrac{e-1}{e}$, we derive a sharper lower bound on the value of improperly allocated agents. Rather than taking these agents' payments as a lower bound on their values, we show that improper winners beating other, higher-valued agents must be facing a similar optimization problem to the rightful winners. Since these improper winners are also bidding higher than the rightful winners, it must be that their values are not too much lower than the values of those who should win. Carefully aggregating this lower bound over all types of all agents yields Theorem~\ref{thm:main}. Notably, our lower bounding argument makes use of the independence assumption on agents' value distributions, unlike the proofs that yielded the previous results.

\subsection{Structure}
In Section~\ref{sec:related}, we discuss our relationship to the literature on first-price auctions and worst-case welfare analysis in mechanism design. We then outline the formal model and technical preliminaries in Section~\ref{sec:prelims}. In the subsequent three sections, we present the proof of our main result. We begin by re-deriving the $\tfrac{e-1}{e}$ bound from \cite{ST13} in a way that will allow for direct comparison to our new approach while developing several technical lemmas that will be useful in proving the tighter welfare bound. In Section~\ref{sec:misalloc}, we present the sharper lower bound on the values of improperly allocated agents. Finally, in Section~\ref{sec:mainproof}, we show how to combine the lower bound in Section~\ref{sec:misalloc} with the approach outlined in Section~\ref{sec:warmup} to prove Theorem~\ref{thm:main}.

\section{Related Work}
\label{sec:related}

The first-price auction presents a daunting obstacle to classical economic analysis because equilibria in the first-price auction are notoriously difficult to compute theoretically. For example, when \cite{V61} first observed that first-price auctions have equilibria which are not welfare-optimal, computing the equilibria of two bidders with asymmetric uniform distributions was posed as an open question. This question was only answered fifty years later in \cite{KZ12}. While other literature in economics has sought to understand equilibria in other special cases (e.g. \cite{plum1992characterization}), a broader view has proved elusive. Indeed, for the related setting of several first-price auctions run simultaneously, \cite{cai2014simultaneous} show that there are computational reasons why such equilibria are challenging to characterize.

To circumvent this obstacle, recent analyses have resorted to worst-case analysis over all equilibria. These approaches allowed the theorist to reason about equilibrium without ever solving for one. The most prevalent tool for this task is \emph{smoothness}, pioneered in \cite{roughgarden2015intrinsic} and adapted for auctions by \cite{S14} and \cite{ST13}. In the case of the first-price auction, this approach has yielded the best known bound of $\tfrac{e-1}{e}$, mentioned in the introduction. This bound also extends to the simultaneous composition of multiple first-price auctions \citep{ST13} when bidders have fractionally subadditive valuations. Smoothness has seen applications far beyond the first price auction; we refer the reader to the excellent survey of \cite{RST16} for a broader picture.

The $\frac{e-1}{e}$ bound however has only been shown to be tight in two situations: the single-item setting when agents' values are correlated \citep{S14}, and in the multi-item case with the simultaneous composition of multiple first-price auctions 
\citet{Christodoulou2016TightBF}.

For revenue, \cite{HHT14} show that the smoothness-based welfare bounds for the first-price auction extend to the objective of revenue as well, assuming the auctioneer has recruited sufficient competition or sets proper reserve prices.

Our approach falls outside the standard smoothness framework. In the space of mechanism design, it is one of few such welfare bounds which are not smoothness-based. We mention two others of note. First, \cite{BL10} consider combinatorial auctions, and prove that mechanisms based on greedy allocation algorithms have equilibria which are approximately welfare-optimal. Second, \cite{christodoulou2015efficiency} consider full-information equilibria of the all-pay auction, and provide tight worst-case welfare guarantees.
\section{Technical Preliminaries}
\label{sec:prelims}

We now lay the formal groundwork for our result. This paper analyzes the \emph{single-item sealed-bid first-price auction}. In such an auction a single item is sold to $n$ agents. Each agent $i$ simultaneously submits a bid $b_i$ to the auctioneer. The agent $i^*$ with the highest bid wins the item, and pays their bid $b_{i^*}$. All other agents pay nothing and win nothing. Let $\tilde x_i(\mathbf b)$ denote the indicator for whether agent $i$ is allocated under bid profile $\mathbf b$, and let $\tilde p_i(\mathbf b)$ denote the payments made by agent $i$ under that same bid profile. Each agent evalutes their allocation and payment using the linear utility function $\tilde u_i(\mathbf b)=v_i\tilde x_i(\mathbf b) - \tilde p_i(\mathbf b)$, where $v_i$ is agent $i$'s value for service, also called their type.

We consider a Bayesian environment, in which each agent $i$'s value is drawn independently from a distribution with CDF $F_i$ and density $f_i$. Note that we do not require agents' value distributions to be identical. We assume agents' values are private, but that the prior distributions are common knowledge. Given a value $v_i$ for agent $i$, we will sometimes refer to $v_i$'s strength in $i$'s distribution by the \emph{quantile} $q_i(v_i)=F_i(v_i)$ of that value.\footnote{Note that this definition of quantile gives strong values high quantiles. This differs from the definition of quantile often used in the literature on revenue maximization in auctions. The reason for this difference is consistency: we discuss quantiles of other random variables, in which it is standard to define quantiles such that high values have high quantile.}

We adopt the standard solution concept of Bayes-Nash equilibrium (BNE). Informally, a BNE is a strategy mapping $b_i(\cdot)$ from values to bids for every agent such that each agent's bid given their value maximizes their expected utility given the strategies of other bidders. Formally, given a profile $\mathbf b(\cdot)$ of bidding strategies for each agent, define the interim allocation probability of agent $i$ bidding $b$ to be $\tilde x_i(b)=\ex{\mathbf v_{-i}}{\tilde x_i(b,\mathbf b_{-i}(\mathbf v_{-i}))}$. Similarly, define the interim expected payments of agent $i$ to be $\tilde p_i(b)=\ex{\mathbf v_{-i}}{\tilde p_i(b,\mathbf b_{-i}(\mathbf v_{-i}))}$. Define the interim expected utility $\tilde u_i(b)$ similarly. A profile of bidding strategies $\mathbf b(\cdot)$ is a BNE if for every agent $i$ with value $v_i$, the following best response inequality holds for every alternate bid $b$: $\tilde u_i(b_i(v_i))\geq \tilde u_i(b)$.

In what follows, we argue assuming agents are bidding according to an arbitrary BNE profile of bidding strategies. Since the strategies map values to bids and bids are mapped to allocation and payments, we will often consider allocations, payments, and utiltities as a function of value, taking the bid functions as implicit. Formally, we will let $x_i(v_i)=\ex{\mathbf v_{-i}}{\tilde x_i(b_i(v_i),\mathbf b_{-i}(\mathbf v_{-i}))}$, $p_i(v_i)=\ex{\mathbf v_{-i}}{\tilde p_i(b_i(v_i),\mathbf b_{-i}(\mathbf v_{-i}))}$, and $u_i(v_i)=\ex{\mathbf v_{-i}}{\tilde u_i(b_i(v_i),\mathbf b_{-i}(\mathbf v_{-i}))}$. Note that we use tildes when the argument to the function is a bid, and omit the tildes when an argument to the function is a value instead.

We study the objective of utilitarian social welfare. The social welfare of a BNE is the expected value of the winner. In other words, $\text{WELF}(\mathbf b(\cdot))=\sum_i v_ix_i(v_i)$. As our benchmark, we compare to the expected value of the bidder with the highest value, i.e. $\ex{\mathbf v}{\max_i v_i}$. This is the welfare of the mechanism which always allocates the highest-valued agent, which can be achieved in equilibrium by a second-price auction. We will state our performance guarantees for the first-price auction as a fraction of this benchmark welfare.

Finally, we note that it will be useful to consider allocation probabilities, expected payments, and expected utility in smaller probability spaces, conditioning, for example, on agent $i$ having the highest value with value $v_i$. Given such an event $\mathcal E$, we will use the shorthand $\tilde x_i(b\,|\,\mathcal E)=\condex{\mathbf v_{-i}}{\tilde x_i(b,\mathbf b_{-i}(\mathbf v_{-i}))}{\mathcal E}$ to denote the allocation probability of agent $i$ given a bid of $b$ conditioned on this event, and so on for payments and utilities.
\section{Warmup: Proving the Standard Bound}
\label{sec:warmup}

Before proving Theorem~\ref{thm:main}, we rederive the $\tfrac{e-1}{e}$ bound of \cite{ST13} in order to highlight the ways in which the proof of our main result differs. In doing so, we will derive several lemmas which will be useful in the proof of the tighter bound. Formally, we will re-prove the following theorem:

\begin{theorem}[\cite{ST13,HHT14}]\label{thm:oldbound}
The expected welfare of any Bayes-Nash equilibrium of the first-price auction with independently-distributed bidders is at least a $\tfrac{e-1}{e}$-fraction of the optimal welfare.
\end{theorem}

To prove Theorem~\ref{thm:oldbound}, we break the welfare of the first-price auction into two quantities: the welfare from agents who win in both the first-price auction and the optimal allocation, and the contribution to welfare from agents who win in the first-price auction, but are not the highest-valued agents. More specifically, we consider the contribution for an arbitrary agent $i$ in the event that agent $i$ has value $v_i$, and that this is the highest value. We denote this event $\ihigh$. Conditioned on $\ihigh$, we will show that the expected value of the winner of the first price auction as a fraction of $v_i$. Formally, we will write the welfare in the following way:

\begin{lemma}\label{lem:welfbreakdown}
	Given any Bayes-Nash equilibrium of the first price auction, let $i^*$ be the random variable given by $i^*=\argmax_{i}\,\,b_i(v_i)$ (breaking ties arbitrarily). The expected welfare in any Bayes-Nash equilibrium of the first-price auction can be written as
	\begin{eqnarray}\label{eq:welfare}
\sum_{i=1}^n \int f_i(v_i)\prob{\ihigh}\Big(v_i x_i(v_i\,|\,\ihigh)+
\condex{}{v_{i^*}}{\ihigh,i^*\neq i}\prob{i^*\neq i\,|\,\ihigh}\Big)\,dv_i.
	\end{eqnarray}
\end{lemma}

Similar to the approach in the smoothness-based proof of \cite{ST13}, we use payments made by the winner to lower bound that agent's contribution to the mechanism's welfare. More formally, given a value profile $\mathbf v$, let $\tau_i(\mathbf v_{-i})$ denote agent $i$'s threshold bid. That is, $\tau_i(\mathbf v_{-i})$ is the bid of the highest bidder other than $i$; $i$ wins if and only if $b_i(v_i)\geq \tau_i(\mathbf v_{-i})$ (modulo tiebreaking). When $i$ loses the auction, the winner pays their bid, which is  $\tau_i(\mathbf v_{-i})$ by definition. Lower bounding agent $i$'s threshold bid will then translate to revenue and hence welfare. We produce such a lower bound with the following lemma:

\begin{lemma}\label{lem:threshquantile}
Let $\tau_i(z,v_i)$ be the threshold bid with quantile $z$ in the distribution of $\tau_i(\mathbf v_{-i})$ conditioned on $\ihigh$. That is, $\tau_i(z,v_i)$ is the value such that the probability that $\tau_i(\mathbf v_{-i})\leq \tau_i(z,v_i)$ is exactly $z$. Then as long as $\tau_i(z,v_i)\geq b_i(v_i)$:
\begin{equation}
\label{eq:quantilebound}
	\tau_i(z,v_i)\geq v_i-\frac{u_i(v_i\,|\,\ihigh)}{z}
\end{equation}
\end{lemma}
\begin{proof}
	We prove the lemma in two steps. First, we prove a best response inequality for $b_i(v_i)$ conditioned on $\ihigh$. While $b_i(v_i)$ may not maximize agent $i$'s utility conditioned on this event, we show that it does yield better utility than any higher bid. We then use this best response inequality to prove (\ref{eq:quantilebound}), by noting that $x_i(v_i\,|\,\ihigh)$ is the CDF of $\tau_i(z,v_i)$, conditioned on $\ihigh$.
	
	To prove a best response inequality conditioned on $\ihigh$, we first compare $\bidutilhigh$ to $\tilde u_i(b)$. In terms of the CDFs $B_j(\cdot)$ and $F_j(\cdot)$, we have that
		\begin{equation}
		\label{eq:utilnocond}
		\bidutil=(v_i-b)\bidalloc=(v_i-b)\prod_{j\neq i}B_j(b).
	\end{equation}
	Meanwhile, conditioned on $i$ having the highest value of $v_i$, we have
	\begin{eqnarray}
		 \notag &&\bidutilhigh=(v_i-b)\bidallochigh\\&&=(v_i-b)\prod_{j\neq i}\prob{b_j(v_j)\leq b\,|\,\ihigh}\notag \\&&=(v_i-b)\prod_{j\neq i}\min\left(\tfrac{B_j(b)}{F_j(v_i)},1\right)\label{eq:utilcond}
	\end{eqnarray}
	It is clear from (\ref{eq:utilnocond}) and (\ref{eq:utilcond}) that the ratio $\tfrac{\bidutilhigh}{\bidutil}$ is decreasing in $b$. But then for any $b>b_i(v_i)$, we have $\tfrac{\bidutilhigh}{\bidutil}\leq \tfrac{\tilde u_i(b_i(v_i)\,|\,\ihigh)}{\tilde u_i(b_i(v_i))}$. Since $b_i(v_i))\geq \bidutil$, this can only be the case if $\tilde u_i(b_i(v_i)\,|\,\ihigh\geq\bidutilhigh$. In other words, for all $b>b_i$,
	\begin{equation}\label{eq:condbr}
	u_i(v_i\,|\,\ihigh)\geq (v_i-b)\bidallochigh
	\end{equation}
	
	We now derive the bound (\ref{eq:quantilebound}). The allocation rule $\bidallochigh$ is the probability that a bid of $b$ is higher than the bids of all agents other than $i$, conditioned on the event $\ihigh$. It follows that $\bidallochigh$ can also be thought of as the CDF (i.e. the quantile function) of this threshold bid, conditioned on $\ihigh$. We may therefore write our conditional best response inequality (\ref{eq:condbr}) as:
	\begin{equation}
u_i(v_i\,|\,\ihigh)\geq (v_i-\tau_i(z,v_i))z.
\end{equation}
Rearranging yields the lemma.
\end{proof}

The lower bound of Lemma~\ref{lem:threshquantile} is stated in terms of the quantile of agent $i$'s threshold bid. We may compute expectations by integrating over quantiles to produce the following lower bound on the value of the winning agent when $i$ loses:

\begin{lemma}\label{lem:easylb}
	For any agent $i$ with value $v_i$, the following inequality holds:
	\begin{eqnarray}
\notag \condex{}{v_{i^*}}{\ihigh,i^*\neq i}\prob{i^*\neq i\,|\,\ihigh}
 \geq v_i(1-\allochigh)+\utilhigh\ln \tfrac{\utilhigh}{v_i}.
	\end{eqnarray}
\end{lemma}

The proof of the lemma requires integrating over all quantiles of agent $i$'s threshold bid, and noting that this threshold bid is a lower bound on the value of $i^*$. We defer the proof to the appendix.

\begin{proof}[Proof of Theorem~\ref{thm:oldbound}]
Plugging the lower bound of Lemma~\ref{lem:easylb} into (\ref{eq:welfare}), we have that the expected welfare from the first-price auction is at least:
\begin{align*}\notag 
\sum_{i=1}^n& \int f_i(v_i)\prob{\ihigh}\Big(\allochigh +
v_i(1-\allochigh)+\utilhigh\ln \tfrac{\utilhigh}{v_i}\Big)\,dv_i\\
&=\sum_{i=1}^n \int f_i(v_i)\prob{\ihigh}\Big(
v_i+\utilhigh\ln \tfrac{\utilhigh}{v_i}\Big)\,dv_i\\
&\geq \sum_{i=1}^n \int f_i(v_i)\prob{\ihigh}\Big(
(1-1/e)v_i \Big)\,dv_i.
\end{align*}
The final line comes from minimizing $v_i+\utilhigh\ln \tfrac{\utilhigh}{v_i}$ over all possible values of $\utilhigh$ between $0$ and $v_i$. We therefore have the following lower bound on the BNE welfare of the first-price auction:
\begin{equation*}
(1-1/e)\sum_{i=1}^n \int f_i(v_i)\prob{\ihigh}v_i\, dv_i.
\end{equation*}
The result follows from noting that the sum in the above expression is exactly the optimal expected welfare.
\end{proof}

\paragraph{Discussion:}

The proof strategy above uses the same principles as other existing welfare guarantees for the first price auction, e.g. \cite{ST13}: it combines a contribution to the welfare of the highest-valued agent when they win with a lower bound on the welfare when they lose. The lower bound on the welfare when $i$ loses takes the form of revenue, which can be lower bounded via a best response argument as above (or in \cite{HHT14}, or via a deviation-based smoothness argument, as in \cite{ST13}.

Both the proof from \cite{ST13} and that of \cite{HHT14} can be formulated in a way that does not require the assumption that agents' values are independently distributed. The proof above, on the other hand, uses the independence assumption in the proof of Lemma~\ref{lem:threshquantile}. The conditioning on the event $\ihigh$ in the welfare breakdown in Lemma~\ref{lem:welfbreakdown} makes this necessary. As we will see below, this conditioning will prove necessary to get a bound better than $\tfrac{e-1}{e}$.
\section{Value Lower Bounds for Misallocated Agents}
\label{sec:misalloc}

The proof of Theorem~\ref{thm:oldbound} broke the welfare into the contribution of the event $\ihigh$ for all agents $i$ and values $v_i$. This welfare contribution was in turn broken into two parts: the value from $i$ when they win in the first-price auction, and the value from the winner $i^*$ when $i$ does not win. In the proof of Theorem~\ref{thm:oldbound}, the value from $i^*$ was lower bounded by that agent's payments. In this section, we take advantage of the independent values assumption to prove a sharper characterization of the value contribution from $i^*$.

To attain this sharper bound, we find symmetry in the bidding problems $i$ and $j$ face through their common opponents: everyone else. Factoring out the symmetric aspects of outbidding the other agents allows us to lowerbound the utility $j$ would receive from deviating to bidding $b_i$ and hence the utility they receive in equilibrium bidding $b_j$. Lemma~\ref{lem:vjlowerbound} captures this lowerbound. The final bound comes then from integrating numerically over all quantiles $q_i$.





We formalize this intuition below.
\begin{lemma}\label{lem:vjlowerbound}
	Let agent $i$ have value $v_i$, quantile $F_i(v_i)=q_i$. Then in any Bayes-Nash equilibrium, for any agent $j$ with value $v_j$, if $b_j(v_j)\geq b_i(v_i)$, then:
	\begin{equation}\label{eq:vjlowerbound}
	v_j\geq v_i \frac{\frac{b_j(v_j)}{v_i} - (1-q_i)\frac{b_j(v_j)}{v_i}\frac{b_i(v_i)}{v_i}-q_i\frac{b_i(v_i)}{v_i} }{1-q_i-\frac{b_i(v_i)}{v_i} + \frac{b_j(v_j)}{v_i}q_i}
	\end{equation}
\end{lemma}
\begin{proof}
In what follows, we suppress the dependence of the best response bids $b_i$ and $b_j$ on $v_i$ and $v_j$ for brevity. Because $i$ and $j$ are bidding in equilibrium, we know that $i$ would prefer not to bid $b_j$, and $j$ would prefer not to bid $b_i$. That is:
\begin{align*}
(v_i-b_i)\tilde x_i(b_i)\geq (v_i-b_j)\tilde x_i(b_j)\\
(v_j-b_j)\tilde x_j(b_j)\geq (v_j-b_i)\tilde x_j(b_i).
\end{align*}
Because agents have independent value distributions, we may write the bid allocation rules $\tilde x_i(\cdot)$ and $\tilde x_j(\cdot)$ in terms of the CDFs of each agent's bid distributions, $B_k(\cdot)$. We have that $\tilde x_i(b)=\prod_{k\neq i} B_k(b)$, and $\tilde x_j(b)=\prod_{k\neq j} B_k(b)$. This allows us to divide the best response inequalities above and cancel $B_k(b)$ for any $k\notin\{i,j\}$. In other words, we have
\begin{equation*}
\frac{(v_i-b_i)}{(v_j-b_i)}\frac{B_j(b_i)}{B_i(b_i)}\geq \frac{(v_i-b_j)}{(v_j-b_j)}\frac{B_j(b_j)}{B_i(b_j)}.
\end{equation*}
Let $q_j=F_j(v_j)$, and $q_i=F_i(v_i)$ be the quantiles of agent $i$ and $j$ in their respective value distributions. Since $b_j\geq b_i$ and best response bids must be increasing in value, $B_j(b_i)\leq B_j(b_j)=q_j$. Since we also have $B_j(b_j)=q_j$, $B_i(b_i)=q_i$, and $B_i(b_j)\leq 1$, we have:
\begin{equation*}
\frac{(v_i-b_i)}{(v_j-b_i)}\frac{q_j}{q_i}\geq \frac{(v_i-b_j)}{(v_j-b_j)}q_j.
\end{equation*}
Cancelling $q_j$ from both sides and rearranging yields the lemma.
\end{proof}

The parameterization by quantile of Lemma ~\ref{lem:vjlowerbound} is particularly important: at $q_i=0$, Equation~\eqref{eq:vjlowerbound} gives $v_j\geq b_j$, exactly the bound used in Theorem~\ref{lem:easylb} to arrive at the $\frac{e-1}{e}$ bound. At quantile $q_i=1$, the bound gives $v_j\geq v_i$.

Lemma~\ref{lem:vjlowerbound} lower bounds $v_j$ in terms of the value $v_i$ of agent $i$, the quantile $q_i$ of agent $i$, and the two agents' bids, $b_i(v_i)$ and $b_j(v_j)$. In what follows, we will use denote this lower bound function by 
\begin{equation*}
\underline v(v_i,q_i,b_i,b_j)\equiv v_i \frac{\frac{b_j}{v_i} - (1-q_i)\frac{b_j}{v_i}\frac{b_i}{v_i}-q_i\frac{b_i}{v_i} }{1-q_i-\frac{b_i}{v_i} + \frac{b_j}{v_i}q_i}.
\end{equation*}
The proof of Theorem~\ref{thm:main} amounts to using $\underline v(v_i,q_i,b_i(v_i),b_{i^*}(v_{i^*}))$ instead of $\tau_i(\mathbf v_{-i})$ as a lower bound on the value of the winner $i^*$ in the event $\ihigh$ when $i^*\neq i$, and combining this with the a lower bound on $b_{i^*}$ obtained from applying Lemma~\ref{lem:threshquantile}.
\section{Proof of Theorem~\ref{thm:main}}
\label{sec:mainproof}
We now integrate the sharper bound of Lemma~\ref{lem:vjlowerbound} into the proof framework outlined in Section~\ref{sec:warmup}. We begin by lowerbounding the value of the winner $v_{i^*}$ when $i$ should win with value $v_i$ but loses.

\begin{lemma}\label{lem:newlb}
	For any agent $i$ with value $v_i$ and quantile $F_i(v_i)=q_i$, the following inequality holds:
	\begin{eqnarray}
	&& \condex{}{v_{i^*}}{\ihigh,i^*\neq i}\prob{i^*\neq i\,|\,\ihigh} \label{eq:misallocbound}
	\\ \notag &&\,\,\,\,\,\,\,\,\,\geq v_i\left(1-\tfrac{\utilhigh}{v_i-b_i(v_i)}\right)-(1-q_i)\utilhigh\ln\left(1+\tfrac{v_i-b_i(v_i)-\utilhigh}{(1-q_i)\utilhigh}\right).
	\end{eqnarray}

\end{lemma}

	The proof of the lemma proceeds in four steps. First, we note that $v_{i^*}\geq \vistarlb$. Second, we show that $\vistarlb$ is increasing in $b_{i^*}$. This allows us to substitute the lower bound on $b_{i^*}$ derived in Lemma~\ref{lem:threshquantile} to bound $\vistarlb$	in terms of the quantile $q_i$ of agent $i$'s value and the quantile of agent $i$'s threshold bid. Integrating over the quantiles of agent $i$'s threshold bid yields the result. We defer the full details to the appendix.

\begin{proof}[Proof of Theorem~\ref{thm:main}]
We have lower bounded the value of the winner $v_{i^*}$ conditioned on $\ihigh$ in terms of $v_i$, $\utilhigh$, and $q_i$. As in the proof of Theorem~\ref{thm:oldbound}, we will minimize this lower bound over all possible choices of $\utilhigh$, holding the other two parameters fixed. We will encounter two challenges in applying this approach. First, to eliminate the dependence of the lower bound on $q_i$, we will need to integrate over all quantiles of agent $i$. Second, there will not be a simple worst-case choice of $\utilhigh$, as there was before. We will nonetheless produce a constant by solving the optimization problem numerically.

We begin by restating our lower bound on the welfare of the first-price auction. The total welfare is 
	\begin{eqnarray}\notag
\sum_{i=1}^n \int f_i(v_i)\prob{\ihigh}\Big(v_i x_i(v_i\,|\,\ihigh)+
\condex{}{v_{i^*}}{\ihigh,i^*\neq i}\prob{i^*\neq i\,|\,\ihigh}\Big)\,dv_i.
\end{eqnarray}
Using the lower bound of Lemma~\ref{lem:newlb}, the fact that $\allochigh=\utilhigh/(v_i-b_i(v_i))$, and algebra, we obtain:
\begin{equation*}
\sum_{i=1}^n \int f_i(v_i)\prob{\ihigh}\Big(v_i-\utilhigh (1-q_i)\ln \left(1+\tfrac{v_i-b_i(v_i)\utilhigh}{(1-q_i)\utilhigh}\right)\Big)\,dv_i.
\end{equation*}
Note that this expression is increasing in $b_i(v_i)$. Taking $b_i(v_i)=0$ and rewriting the welfare in terms of agent $i$'s quantile $q_i$, we obtain the lower bound:
\begin{equation*}
\sum_{i=1}^n \int_0^1 \prob{\mathcal E_i(v_i(q_i))}v_i(q_i)\left(1-r_i(q_i)(1-q_i)\ln \left( 1+\frac{1-r_i(q_i)}{(1-q_i)r_i(q_i)} \right)\right)\,dq_i,
\end{equation*}
where $r_i(q_i)=\utilqhigh/v_i(q_i)$. Note that the minimum of 
\begin{equation*}
1-r_i(q_i)(1-q_i)\ln \left( 1+\frac{1-r_i(q_i)}{(1-q_i)r_i(q_i)} \right)
\end{equation*}
over $r_i(q_i)\in[0,1]$ now depends on $q_i$. Define 
\begin{equation*}
\ell(q_i)=\min_{r_i(q_i)\in[0,1]}\left[1-r_i(q_i)(1-q_i)\ln \left( 1+\frac{1-r_i(q_i)}{(1-q_i)r_i(q_i)} \right)\right]
\end{equation*}
We have may lower bound the welfare of the first price auction as:
\begin{equation*}
\sum_{i=1}^n \int_0^1 \prob{\mathcal E_i(v_i(q_i))}v_i(q_i)\ell(q_i)\,dq_i.
\end{equation*}

Finally, let $\gamma_i(q_i)=\prob{\mathcal E_i(v_i(q_i))}v_i(q_i)$. We may lower bound the equilibrium welfare as:
\begin{align*}
\sum_{i=1}^n \int_0^1 \gamma_i(q_i)\ell(q_i)\,dq_i&= \sum_{i=1}^n \int_0^1 \left(\gamma_i'(q_i)\int_{q_i}^1\ell(t)\,dt\right)\,dq_i-\gamma(s)\int_{s}^1\ell(t)\,dt\bigg\rvert_{s=0}^1\\&=\sum_{i=1}^n \int_0^1 \gamma_i'(q_i)\int_{q_i}^1\ell(t)\,dt\,dq_i\\
&\geq \sum_{i=1}^n \int_0^1 \gamma_i'(q_i)\left(\min_x\frac{\int_{x}^1\ell(t)\,dt}{1-x}\right)(1-q_i)\,dq_i\\
&=\left(\min_x\frac{\int_{x}^1\ell(t)\,dt}{1-x}\right)\sum_{i=1}^n \int_0^1 \gamma_i'(q_i)\int_{q_i}^11\, dt\,dq_i\\
&=\left(\min_x\frac{\int_{x}^1\ell(t)\,dt}{1-x}\right)\sum_{i=1}^n \int_0^1 \gamma_i(q_i)\,dq_i.
\end{align*}
The first line follows from integration by parts. The second line follows from the fact that $
\gamma(0)=0$. The fourth line follows again from integration by parts. Since $\sum_{i=1}^n \int_0^1 \gamma_i(q_i)\,dq_i$ is equal to the expected optimal welfare, all that remains is to compute 
\begin{equation*}
\min_x\frac{\int_{x}^1\ell(t)\,dt}{1-x},
\end{equation*}
which will in turn yield our welfare bound. Doing so numerically yields the value stated in the theorem.
\end{proof}
\section{Acknowledgments}
This work was supported by the Shanghai Sailing Program (Grant No. 18YF1407900).
\bibliographystyle{apalike}
\bibliography{bibs}

\appendix
\section{Deferred Proofs}

\begin{proof}[Proof of Lemma~\ref{lem:easylb}]
	Consider a realization of $\mathbf v$ in which $i$ does not win. This means that $i$'s bid is less than its threshold bid $\tau_i(\mathbf v_{-i})$, while the bid of the winner, $i^*$, is at least $\tau_i(\mathbf v_{-i})$. Since in equilibrium agents don't overbid, we have that $v_{i^*}\geq \tau_i(\mathbf v_{-i})$. We therefore have that:
	\begin{equation}\label{eq:easylb}
	\condex{}{v_{i^*}}{\ihigh,i^*\neq i}\prob{i^*\neq i\,|\,\ihigh}\geq \condex{}{\tau_i(\mathbf v_{-i})}{\ihigh,i^*\neq i}\prob{i^*\neq i\,|\,\ihigh}.
	\end{equation}
	To evaluate the quantity on the right-hand side of (\ref{eq:easylb}), we will integrate over the quantile space of agent $i$'s threshold bid $\tau_i(\mathbf v_{-i})$. We obtain the following sequence of inequalities:
	\begin{align}
	\notag \condex{}{\tau_i(\mathbf v_{-i})}{\ihigh,i^*\neq i}\prob{i^*\neq i\,|\,\ihigh}&=\int_{\allochigh}^1\tau_i(z,v_i)\,dz\\
	\notag &\geq \int_{\allochigh}^1 v_i-\frac{\utilhigh}{z}\,dz\\
    \notag &= v_i(1-\allochigh)+\utilhigh\log(\allochigh) \\
    \notag &\geq v_i(1-\allochigh)+\utilhigh\log\left(\frac{\utilhigh}{v_i}\right).
	\end{align}
	The second line follows from the lower bound of Lemma~\ref{lem:threshquantile}, the final from noting $\frac{\utilhigh}{v_i} \leq \allochigh$.
\end{proof}

\begin{proof}[Proof of Lemma~\ref{lem:newlb}]
	
	By definition, $v_{i^*}$ always bids higher than bidder $i$ with value $v_i$, in the event that $i\neq i^*$. We may therefore begin by applying Lemma~\ref{lem:vjlowerbound}, which states that $v_{i^*}\geq \vistarlb$. Hence, we may lower bound the right-hand side of (\ref{eq:misallocbound}) by 
	
	\begin{equation*}\condex{}{\vistarlb}{\ihigh,i^*\neq i}\prob{i^*\neq i\,|\,\ihigh}.\end{equation*}
	
	Next, we show that $\tfrac{\partial \vistarlb}{\partial b_{i^*}}\geq 0$. Computing the derivative and rearranging it yields that
	\begin{equation*}
	\frac{\partial \vistarlb}{\partial b_{i^*}}=\frac{(1-q_i)(1-\tfrac{b_i}{v_i})^2}{(1-q_i-\tfrac{b_i}{v_i}+\tfrac{b_j}{v_i}q_i)^2}\geq 0.
	\end{equation*}
	Note that as long as $i^*$ is outbidding agent $i$, it must be that $b_{i^*}\geq \tau_i(\mathbf v_{-i})$. By Lemma~\ref{lem:threshquantile} we therefore have $b_{i^*}\geq v_i-\tfrac{\utilhigh}{z}$, where $z$ is the quantile of $\tau_i(\mathbf v_{-i})$ in its distribution conditioned on the event $\ihigh$.

	Hence, we may write 
	\begin{align}
	\notag	&\condex{}{\vistarlb}{\ihigh,i^*\neq i}\prob{i^*\neq i\,|\,\ihigh}\\
	\label{eq:lbwithquant}&\geq \condex{}{\underline v\left (v_i,q_i,b_i,v_i-\tfrac{\utilhigh}{z}\right )}{\ihigh,i^*\neq i}\prob{i^*\neq i\,|\,\ihigh}.
	\end{align}
	
	As in the proof of Lemma~\ref{lem:easylb}, we may compute the right-hand side of (\ref{eq:lbwithquant}) by integrating over the quantiles of $\tau_i(\mathbf v_{-i})$. That is, $\condex{}{\vistarlb}{\ihigh,i^*\neq i}\prob{i^*\neq i\,|\,\ihigh}$ is at least 
	\begin{equation*}
	\int_{\tfrac{\utilhigh}{v_i-b_i}}^1\underline v\left (v_i,q_i,b_i,v_i-\tfrac{\utilhigh}{z}\right )\,dz.
	\end{equation*}
	Evaluating the integral yields the inequality stated in the lemma.
\end{proof} 
\end{document}